\documentclass[11pt]{article}
\usepackage[T1]{fontenc}
\usepackage[margin=1in]{geometry} 
\usepackage{libertineRoman}
\usepackage{biolinum}
\usepackage{libertineMono} 
\usepackage[latin9]{inputenc}

\usepackage{cite}
\usepackage{enumitem}
\usepackage{pgfplots}
\pgfplotsset{compat=1.18} 

\usepackage{algorithm}
\usepackage{algpseudocode}
\usepackage{amsfonts}
\usepackage{amsmath}
\usepackage{amssymb}
\usepackage{amsthm}
\usepackage{authblk}
\usepackage{bbm}
\usepackage{enumerate}
\usepackage{graphicx}
\usepackage{ifthen}
\usepackage{latexsym}
\usepackage{mathtools}
\usepackage{multicol}
\usepackage{palatino}
\usepackage{mathpazo}

\usepackage{thmtools} 
\usepackage{tabto}
\usepackage[colorlinks=true,linkcolor=blue,anchorcolor=blue,citecolor=red,urlcolor=magenta]{hyperref}
\usepackage[capitalize]{cleveref} 

\crefname{appendix}{Appendix}{Appendices}
\Crefname{appendix}{Appendix}{Appendices}

\usepackage{color, xcolor}
\usepackage{stmaryrd} 
\usepackage{makeidx} 

\crefname{equation}{}{}

\newtheorem{theorem}{Theorem}

\newtheorem{definition}[theorem]{Definition}
\newtheorem{lemma}[theorem]{Lemma} 
\newtheorem{corollary}[theorem]{Corollary}
\newtheorem{remark}[theorem]{Remark}
\newtheorem{proposition}[theorem]{Proposition}

\newtheorem{example}[theorem]{Example} 
\newtheorem{fact}[theorem]{Fact}  

\newcommand{\tinyspace}{\mspace{1mu}}

\newcommand{\tr}{\operatorname{Tr}}

\newcommand{\bra}[1]{\langle #1 |}
\newcommand{\ket}[1]{| #1 \rangle}
\newcommand{\braket}[2]{\langle #1 | #2 \rangle}

\newcommand\ip[2]{\ensuremath{\langle#1 , #2\rangle}} 

\newcommand\Tr{\mathop{\rm Tr}\nolimits} 
\newcommand\range{\mathop{\rm range}\nolimits}

\newcommand{\norm}[1]{\left\lVert\tinyspace#1\tinyspace\right\rVert}

\newcommand{\Ik}{\mathcal{I}_k^n}

\def\C{\mathbb{C}}

\def\Pos{\operatorname{Pos}}

\def\Herm{\operatorname{Herm}}

\newcommand{\kb}[1]{\ket{#1} \bra{#1}}
 
\newcommand{\problem}[1]{{\textsc{#1}}} 
\newcommand{\poly}{\mathrm{poly}}  
\newcommand{\lw}{{\problem{LearningWidth}}}
\newcommand{\mlw}{{\problem{MinLearningWidth}}} 
\newcommand{\fw}{{\problem{FactorWidth}}}  
\newcommand{\mfw}{{\problem{MinFactorWidth}}}   

\newcommand{\WOPT}{\mathrm{WeakOPT}} 
\newcommand{\WMEM}{\mathrm{WeakMEM}} 

\title{\vspace{-1cm}The complexity of perfect quantum state classification} 

\author{Nathaniel Johnston\thanks{Department of Mathematics \& Computer Science,
Mount Allison University. \\ \tt{njohnston@mta.ca},
\url{https://njohnston.ca/}}, \; 
Benjamin Lovitz\thanks{Department of Computer Science and Software Engineering,
Concordia University. \\ \tt{benjamin.lovitz@concordia.ca},
\url{https://www.benjaminlovitz.com/}}, \; 
Vincent Russo\thanks{Unitary Foundation. \\ \tt{vincent@unitary.foundation},
\url{https://vprusso.github.io/}}, \; 
Jamie Sikora\thanks{Department of Computer Science, Virginia Tech. \\ \tt{sikora@vt.edu}, \url{https://sites.google.com/site/jamiesikora/}} 
} 

\date{October 24, 2025} 

\begin{document}

\maketitle 

\begin{abstract} 
The problem of quantum state classification asks how accurately one can identify an unknown quantum state that is promised to be drawn from a known set of pure states. 
In this work, we introduce the notion of $k$-\emph{learnability}, which captures the ability to identify the correct state using at most $k$ guesses, with zero error. 
We show that deciding whether a given family of states is $k$-learnable can be solved via semidefinite programming. 
When there are $n$ states, we present polynomial-time (in $n$) algorithms for determining $k$-learnability for two cases: when $k$ is a fixed constant or the dimension of the states is a fixed constant. 
When both $k$ and the dimension of the states are part of the input, we prove that there exist succinct certificates placing the problem in NP, and we establish NP-hardness by a reduction from the classical $k$-clique problem. 
Together, our findings delineate the boundary between efficiently solvable and intractable instances of quantum state classification in the perfect (zero-error) regime. 
\end{abstract} 


\section{Introduction}

Quantum state distinguishability is a cornerstone problem in quantum information
theory~\cite{watrous2018theory, nielsen2000quantum, wilde2013quantum, helstrom1969quantum, belavkin1975optimal, belavkin1975optimal2, deconinck2010qubit, virmani2001optimal, walgate2000local}
that addresses a fundamental question: \textit{Given a quantum system prepared in one of several known states, can we identify the given state?}
This problem has deep connections to quantum communication, cryptography, and computational
complexity~\cite{bae2015quantum, barnett2009quantum, chefles2000quantum,
bergou2010discrimination}.

In the setting of quantum state exclusion~\cite{bandyopadhyay2014conclusive}, the goal is to determine which state the system is \emph{not} in, rather than which state it \emph{is}. 
Perfect quantum state exclusion (also called antidistinguishability~\cite{caves2002conditions, havlivcek2020simple, heinosaari2018antidistinguishability, mishra2023optimal, russo2023inner,
johnston2025tight}) of a state from a known set is possible when there exists a measurement that perfectly \emph{rules out} a state that was sent.   
This exclusion perspective has proven fruitful for understanding nonlocal properties of quantum states~\cite{leifer2014quantum, leifer2020noncontextuality, pusey2012reality, molina2019povms} and has applications in quantum cryptographic protocols~\cite{crickmore2020unambiguous}. 
One noteworthy point (that we leverage later in this paper) is that the notion of antidistinguishability and the property of $k$-\emph{incoherence} of a matrix~\cite{ringbauer2018certification} are intricately
linked~\cite{johnston2025tight}.

This work explores intermediate scenarios between perfect distinguishability and perfect exclusion. 
Rather than asking whether we can identify or rule out states with certainty, we ask: \textit{How many guesses are
needed to guarantee correct identification?} 
This perspective naturally leads to a new task which we call $k$-\emph{learnability}: an ensemble of states is $k$-\emph{learnable} if there exists a measurement that, upon obtaining an outcome, allows us to narrow down the identity of the unknown state to at most $k$ possibilities, with the guarantee that one of these $k$ guesses is correct.

The $k$-learnability problem has natural connections to quantum machine learning~\cite{schuld2021machine, wittek2014quantum},  quantum learning theory~\cite{arunachalam2017guest, anshu2024survey}, and tomography~\cite{james2001measurement, gross2010quantum}.  
In these fields, the problem of learning quantum states from measurement data has received considerable attention, with complexity measures often tied to the number of measurements or samples required. The $k$-learnability framework reinterprets learnability: it measures the number of distinct candidate states that must be considered, rather than the number of identical samples required.


\subsection{Perfectly distinguishability and learning width}   

An indexed list of (pure) states $\{ \ket{\psi_1}, \ldots, \ket{\psi_n} \}$ is called \emph{perfectly distinguishable} if there exists a POVM that, when applied to an unknown state drawn from the set, always perfectly guesses the index (that is, without error).  
It is well-known that a set of quantum states is perfectly distinguishable if and only if the states are pairwise orthogonal. 
However, if we relax this definition, one can, perhaps surprisingly, learn other aspects of this set without error, even when the states have non-trivial overlap. 

The task of perfect state exclusion (also called antidistinguishability) asks when it is possible to perfectly \textit{exclude} an index. 
A popular example where this can be achieved is the so-called trine states~\cite{holevo1973information, peres1991optimal} 
\begin{equation} 
    \ket{\psi_1} = \ket{0}, 
    \quad 
    \ket{\psi_2} = -\frac{1}{2} \left( \ket{0} + \sqrt{3} \ket{1} \right), 
    \quad 
    \ket{\psi_3} = -\frac{1}{2} \left( \ket{0} - \sqrt{3} \ket{1} \right).   
\end{equation}  
Optimal bounds for this task are known~\cite{bandyopadhyay2014conclusive,
johnston2025tight} noting that one need not require pairwise orthogonality for a
set of states to be antidistinguishable.  

To illustrate $k$-learnability, consider the following generalization of the
trine states to four states in three dimensions, called the \emph{tetrahedral
states}: 
\begin{equation}
    \begin{aligned}
        \ket{\psi_1} = \frac{1}{\sqrt{3}} (\ket{0} + \ket{1} + \ket{2}), \quad & 
        \ket{\psi_2} = \frac{1}{\sqrt{3}} (\ket{0} - \ket{1} - \ket{2}), \\ 
        \ket{\psi_3} = \frac{1}{\sqrt{3}} (-\ket{0} - \ket{1} + \ket{2}), \quad & 
        \ket{\psi_4} = \frac{1}{\sqrt{3}} (-\ket{0} + \ket{1} - \ket{2}). 
    \end{aligned}
\end{equation}

These states are $2$-\textit{learnable}: upon receiving one of them, one can always guess two states from which it was selected without error. 
This can be accomplished using the POVM $M_{i,j} = \frac{1}{2}
\kb{\phi_{i,j}}$, where
\begin{equation}
    \begin{aligned}
    \ket{\phi_{1,2}} &= \frac{1}{\sqrt{2}}(\ket{1} + \ket{2}), \quad
    \ket{\phi_{1,3}} = \frac{1}{\sqrt{2}}(\ket{0} + \ket{1}), \quad
    \ket{\phi_{1,4}} = \frac{1}{\sqrt{2}}(\ket{0} + \ket{2}), \\
    \ket{\phi_{2,3}} &= \frac{1}{\sqrt{2}}(\ket{0} - \ket{2}), \quad
    \ket{\phi_{2,4}} = \frac{1}{\sqrt{2}}(\ket{0} - \ket{1}), \quad
    \ket{\phi_{3,4}} = \frac{1}{\sqrt{2}}(\ket{1} - \ket{2}).
    \end{aligned}
\end{equation} 
We now formally define the notion of $k$-learnability. 

\begin{definition}\label{defn:k_learn} 
    Let $k$ be a positive integer. 
    An indexed list of (possibly sub-normalized) pure quantum states $\{ \ket{\psi_1}, \ldots, \ket{\psi_n} \} \subset \C^d$ is $k$-\emph{learnable} if there exists a POVM 
    \begin{equation} 
    \{ M_S : S \subseteq \{ 1, \ldots, n \}, |S| = k\}, 
    \end{equation} 
    such that $\bra{\psi_i} M_S \ket{\psi_i} = 0$ whenever $i \notin S$. 
    The \emph{learning width} of the list is the smallest $k$ such that it is $k$-learnable. 
\end{definition} 

\begin{remark} 
We note that later, when we discuss the decision problem of determining the $k$-learnability of pure states, we need to define \emph{approximate} $k$-learnability. 
For this, it is convenient to define the $k$-learnability of \emph{subnormalized pure states}. 
Notice that as long as the states have positive norm, the task of determining if they are $k$-learnable is the same; it is only a matter of finding a POVM with the correct indices and orthogonality conditions. 
\end{remark} 

Note that $1$-learnability is equivalent to perfect distinguishability and $(n-1)$-learnability is equivalent to antidistinguishability. A set has learning width $1$ if and only if the states are pair-wise orthogonal. 
Every set has learning width at most $n$, trivially, and a set is antidistinguishable if and only if it has learning width at most $n-1$. 
The trine states and the tetrahedral states both have learning width $2$ since they are $2$-learnable but not $1$-learnable. 

We can trivially construct indexed lists of pure states that are $k$-learnable for any $n\geq k$ in the following way. 
Consider the indexed list of states 
\begin{equation} \label{ex:funny}
\{ \underbrace{\ket{0}, \ket{0}, \ldots, \ket{0}}_{k \text{ copies}}, \underbrace{\ket{1}, \ket{1}, \ldots, \ket{1}}_{k \text{ copies}}, \ket{2}, \ket{2}, \ldots \}. 
\end{equation} 
Measuring in the computational basis will reveal which $k$-subset the states were in. 
While this example is rather straightforward, we note that it also illustrates why we use the term \emph{indexed list} of states instead of \emph{set}, as we are guessing the index, not the state. 
Indeed, with the states in Equation~\cref{ex:funny}, one can always guess the state, but not necessarily the index. 


\subsection{Our results} 

We establish both positive and negative results concerning the computational complexity of deciding $k$-learnability. Specifically, we show that when certain parameters are fixed, the problem admits efficient algorithms: the $k$-learnability of a set of $n$ states can be decided in $\operatorname{poly}(n)$ time when either the number of guesses $k$ is fixed (\cref{thm:informal_k_fixed}) or when the dimension of the underlying Hilbert space is fixed (\cref{thm:learning_easy_small_dim_informal}). 
In these regimes, the structure of the problem allows reductions to semidefinite programs of bounded size, which can be solved efficiently using standard approximation techniques. 
In contrast, when both $k$ and the dimension are treated as part of the input, the problem becomes substantially harder: we prove that it is NP-complete by establishing membership in NP and constructing a reduction from the classical $k$-clique problem (\cref{thm:NP_informal}, \cref{thm:nphard_informal}). 
Together, these results delineate the boundary between the tractable and intractable instances of quantum state classification. 
In what follows, we state these results informally 
and  provide intuition for the underlying proofs. 

The starting point for our proofs is the observation that an indexed list of states $\{ \ket{\psi_1}, \ldots, \ket{\psi_n} \}$ is $k$-learnable if and only if its associated Gram matrix $G_{ij} = \braket{\psi_i}{\psi_j}$ possesses a particular structural property called $k$-\emph{incoherence}~\cite{ringbauer2018certification}. 
Intuitively, this means that the Gram matrix can be expressed as a convex combination of outer products of vectors that are supported on at most $k$ coordinates.  
In the following sections (see~\cref{ssect:gram}), we formalize this correspondence and show how it serves as the foundation for both our algorithmic and hardness results. 

\begin{theorem}[Informal, see~\cref{sect:ellipsoid}]\label{thm:informal_k_fixed}
Given an indexed list of (possibly sub-normalized) pure quantum states $\{ \ket{\psi_1}, \ldots, \ket{\psi_n} \} \subset \C^n$, 
there exists a polynomial-time algorithm (in $n$) to test if the list is $k$-learnable when $k$ is a fixed constant. 
\end{theorem} 

The proof proceeds by stating the problem as a semidefinite program in which there are polynomially-bounded variable size and also polynomially-bounded number of constraints. 
Then using standard algorithms for efficiently approximating semidefinite programs, e.g., the Ellipsoid Method~\cite{grotschel2012geometric}, this gives the result. 

The next theorem gives another regime where we can efficiently test $k$-learnability.

\begin{theorem}[Informal, see~\cref{cor:learning_easy_small_dim}]\label{thm:learning_easy_small_dim_informal}
Given an indexed list of (possibly sub-normalized) pure quantum states $\{ \ket{\psi_1}, \ldots, \ket{\psi_n} \} \subset \C^d$, 
there exists a polynomial-time algorithm to test if the list is $k$-learnable when $k$ is now part of the input, but the dimension $d$ is fixed. 
\end{theorem} 

Translating to $k$-incoherence, this result says that it can be determined in polynomial time whether the (constant) rank $d$ Gram matrix is $k$-incoherent. This is clear when $d=1$ by extremality of pure states. In general, one can formulate the question of $k$-incoherence as a semidefinite program in $\sim\binom{n}{k}$ matrix variables, which could be exponential in $n$, e.g., if $k \approx n/2$. The crux of the proof is the observation that for any constant rank, one can efficiently discard all but $\operatorname{poly}(n)$ of these matrix variables.

When the number of allowed guesses $k$ is treated as part of the input, and the Hilbert space dimension is also allowed to scale with the problem size, the situation changes dramatically. 
In this more general setting, the combinatorial structure of the $k$-learnability condition becomes significantly richer, and the semidefinite formulations that are efficient in the fixed-parameter cases no longer suffice to yield a polynomial-time algorithm. 
Intuitively, as both $k$ and the dimension grow, the space of possible candidate decompositions of the Gram matrix expands exponentially, making it (seemingly) computationally infeasible to solve via convex optimization. 
In fact, we show that the problem is unlikely to be tractable by showing that is it NP-complete. 
This is established through two complementary results, one demonstrating that $k$-learnability lies in NP via succinct proofs, and another proving NP-hardness by reduction from the classical $k$-clique problem. 
Together, these theorems provide a tight characterization of the problem's computational difficulty in the general case. 

\begin{theorem}[Informal, see~\cref{sect:inNP}]\label{thm:NP_informal}
The decision variant of $k$-learnability is in NP. 
\end{theorem} 

We introduce the decision variant of the problem in~\cref{sect:dec}.  

Notice that a proof of $k$-learnability could be to simply non-deterministically guess the POVM operators in the learning POVM. 
However, there are $\binom{n}{k}$ such operators, which is exponentially large when $k \approx n/2$ (as seen by Stirling's approximation, for example), thus a different method is required. 
To do this, we reformulate the problem in terms of $k$-incoherence, and use the geometry of that set along with Carath\'{e}odory's theorem. 
This shows the existence of succinct proofs. 

\begin{theorem}[Informal, see~\cref{thm:nphard}]\label{thm:nphard_informal}
The decision variant of $k$-learnability is NP-hard. 
\end{theorem} 

We prove this by a series of reductions. First, we prove a formal equivalence to the decision variant of $k$-incoherence (see~\cref{thm:reducible}). We then prove a reduction from optimization over $k$-incoherent matrices, and finally a reduction from this optimization problem to the $k$-\problem{Clique} problem for graphs, which is NP-hard.

\paragraph{Author note.} 
\cref{thm:nphard_informal} was also obtained independently and concurrently in~\cite{harrow2025randomized}. We note that their work focused on $k$-incoherence, and they did not observe the equivalence to state classification proven in the present work. 


\subsection{Structure of the paper} 

In the next section, we establish a fundamental connection between the $k$-learnability of quantum states and the notion of $k$-incoherence of their Gram matrix. This provides the analytical foundation for studying the classification problem.~\cref{sect:dec} formalizes both the decision and optimization variants of learnability by framing them as weak membership problems over a suitably defined convex set. 
We also introduce the weak optimization version of this problem and, in~\cref{sec:hardness-proofs}, present rigorous proofs showing that it can be used to solve the $k$-clique problem. 
The decision problem is shown to lie in NP in~\cref{sect:inNP}, through an analysis of the geometry of the set of $k$-incoherent matrices. Efficient algorithms for several special cases are then discussed in~\cref{sect:ellipsoid} and~\cref{sec:low_rank}. 
Finally, we provide links to open-source software that implements the methods developed in this work.


\section{How to tell if a set of states is \texorpdfstring{$k$}{k}-learnable?}\label{sec:how}

The main focus of this work is to study when a given set of states is $k$-learnable. 
We address this from both an analytical and computational complexity point of view.  

To start, we model this as a semidefinite program (SDP). 
Suppose we are given an indexed list of states $\{\ket{\psi_1}, \ldots, \ket{\psi_n} \}\subset \mathbb{C}^d$ 
and a POVM $\{ M_S: S \subseteq \{ 1, \ldots, n \}, |S| = k \}$ which is intended to classify a given state into a subset of size $k$. 
Then, the average error this POVM makes is equal to the optimal objective function value of the following SDP  
    \begin{equation} \label{SDP_OG}
        \min\left\{ \frac{1}{n} \sum_{i=1}^n \bra{\psi_i} \left(\sum_{S \, : \, i \notin S} M_S\right) \ket{\psi_i} : M_S \in \Pos(\C^d), \sum_{\substack{S \subseteq \{ 1, \ldots, n \}\\ |S| = k}} M_S = I \right\} 
    \end{equation} 
where we denote the set of $n \times n$ positive semidefinite (resp. Hermitian) matrices as $\Pos(\C^n)$ (resp. $\Herm(\C^n)$). 
(Note that the minimum is attained since we are optimizing a continuous function over a compact set.) 
Thus, to determine if a list is $k$-learnable, one needs only to determine if the optimal value of the above SDP is $0$ or not. 
However, special care must be taken when posing this question in the context of finite-precision algorithms. 

When one examines the SDP above, there are a few things that are noteworthy when considering how difficult it is to solve. 
Often, one can efficiently approximate semidefinite programs with respect to their variable size and number of constraints, but the above SDP has $\binom{n}{k}$ matrix variables, each of size $n \times n$, which could be exponential when, e.g., $k \approx n/2$, as seen by Stirling's approximation.  
This suggests that this semidefinite program could be \emph{hard} to approximate with respect to $n$ and $k$ in general.  


\subsection{Gram matrices and factor width} 
\label{ssect:gram}

It turns out that determining whether states are $k$-learnable reduces to asking whether certain matrices can be decomposed as sums of positive semidefinite matrices with restricted support.  
To see this, consider the following SDP
\begin{equation} \label{SDP_red} 
\min \left\{ \sum_{i=1}^n \bra{i} \sum_{S: i \notin S} W_S \ket{i} : W_S \in \Pos(\C^n), \; \sum_{S \subseteq \{ 1, \ldots, n \}, |S| = k} W_S = \frac{1}{n} G \right\} 
\end{equation} 
where $G$ is the Gram matrix, i.e., $G_{i,j} = \braket{\psi_i}{\psi_j}$. 

This SDP ends up being useful to state the decision problem of classification, as seen in the following Lemma (proved in~\cref{sdp_proof}). 

\begin{lemma} \label{lem:SDP_same} 
Let $k$ be a positive integer and let $\{ \ket{\psi_1}, \ldots, \ket{\psi_n} \} \subset \C^d$ be an indexed list of (possibly sub-normalized) pure quantum states. 
Then SDP~\cref{SDP_OG} and SDP~\cref{SDP_red} have the same optimal value. 
\end{lemma} 

Note that we are ultimately interested in whether either SDP has an optimal value $0$. 
In this case, we have $\bra{i} W_S \ket{i} = 0$ for each $i$ and each $S$ such that $i \notin S$. 
For such a pair $(i, S)$, this means that the $(i, i)$-th diagonal entry of $W_S$ equals $0$ and since $W_S$ is positive semidefinite, the entire $i$-th row and column must be $0$. 
Since this holds for all $i \notin S$, there must be $n-k$ rows and columns that are all $0$. 
This aligns exactly with the notion of generalized incoherences, discussed next.  

In the theory of quantum resources, $k$-incoherence~\cite{johnston2022absolutely, levi2014quantitative,ringbauer2018certification, sperling2015convex} plays the role of quantifying how much superposition is present in a (mixed) quantum state. 
This particular interpretation of it is not relevant for our purposes, so we just recall its mathematical definition. 

\begin{definition} \label{defn:k_incoh} 
Let $k$ be a positive integer. 
A positive semidefinite matrix $X$ is called $k$-\emph{incoherent} if there exists a positive integer $m$ and a set of vectors $v_i$ with the property that each $v_i$ has at most $k$ non-zero entries, such that 
\begin{equation}\label{eq:k_incoh}
X = \sum_{i=1}^m v_i v_i^*. 
\end{equation} 
\end{definition}  

With this defined, we can now state the following, which is a generalization of a result from~\cite{johnston2025tight} from $n-1$ to general $k$ and defer the proof to~\cref{app:equiv}.  

\begin{lemma} \label{lem:obvious}
Let $k$ be a positive integer. 
An indexed list of (possibly sub-normalized) pure quantum states $\{ \ket{\psi_1}, \ldots, \ket{\psi_n} \} \subset \C^d$ is $k$-\emph{learnable} if and only if its Gram matrix is $k$-incoherent.  
\end{lemma} 

Another closely related definition is that of factor width~\cite{barioli2003maximal, johnston2025factor, boman2005factor}, below. 

\begin{definition}\label{defn:factor_width}
    Let $k$ be a positive integer. The \emph{factor width} of a positive semidefinite
    matrix $X$ is the smallest $k$ such that it is $k$-incoherent. 
\end{definition} 

Equipped with this definition, we have the immediate corollary. 

\begin{corollary}
Let $k$ be a positive integer. 
An indexed list of (possibly subnormalized) pure quantum states $\{ \ket{\psi_1}, \ldots, \ket{\psi_n} \} \subset \C^d$ has learning width equal to the factor width of its Gram matrix.   
\end{corollary}  

With this direct relationship between the learning width of pure states and the factor width of its Gram matrix, we now define the decision and optimization versions of these problems, and show that they are polynomial-time reducible to each other. 


\section{Weak membership and optimization problems} 
\label{sect:dec}

In order to discuss the computational complexity of $k$-incoherence, we consider it as a weak membership problem. 
We adopt the notation and definitions from~\cite{liu2007complexity} and apply it to the set of trace-bounded $k$-incoherent matrices, denoted by 
\begin{equation} 
\Ik := \{ X \in \Herm(\C^n) \; : \; \tr(X) \leq 1, X \text{ is } k\text{-incoherent} \}. 
\end{equation} 
This is clearly a closed, convex set, and in Lemma~\ref{lem:abs_k_incoh_ball_here} we prove that it has non-empty interior.

We define the following three sets of Hermitian matrices that are relevant to weak membership of sets with respect to the Frobenius norm defined as $\| A \|_F = \sqrt{\tr(A^*A)}$.  
\begin{align} 
B(X,\delta) & := \{ Y \; : \; \| X - Y \|_{F} \leq \delta \} \label{normball} \\ 
S(\Ik, \delta) & := \{ X \; : \; \exists Y \in \Ik \; \text{ such that } \, \| X - Y \|_{F} \leq \delta \} \label{setball} \\ 
S(\Ik, -\delta) & := \{ X \; : \; B(X, \delta) \subseteq \Ik \} \label{strictinterior}
\end{align} 

The set \cref{normball} is the ball centered at $X$ of radius $\delta$, the set \cref{setball} is a $\delta$-outer approximation of the set $\Ik$, and the set \cref{strictinterior} is a $\delta$-inner approximation of the set $\Ik$. 
Decision problems are typically formulated relative to these sets in order to avoid spurious computational difficulty associated with points close to the boundary.  

We can now define the weak membership problem. 

\begin{definition}[$\WMEM(K, \delta)$]  
The weak membership problem over a set $K$ is as follows. 
Given a Hermitian matrix $X$ and a real parameter $\delta > 0$, decide between the following two cases: 
\begin{itemize}[itemsep=2pt, topsep=8pt]
\item \textup{Yes case:} $X \in S(K,-\delta)$;  
\item \textup{No case:} $X \in S(K, \delta)$. 
\end{itemize} 
\end{definition} 

From this, we can define our first decision problem. 

\begin{definition}[\fw] \label{defn:kinc}
Given a positive integer $k$, $X \in \Herm(\C^n)$, and real parameter $\delta > 0$, solve the weak membership problem $\WMEM(\Ik, \delta)$. 
\end{definition} 

Before defining the weak membership problem for $k$-learnability, we take a closer look at the possible Gram matrices arising from a list of pure states. 
Consider the list $\{ \ket{\psi_1}, \ldots, \ket{\psi_n} \}$ and its Gram matrix $G_{i,j} = \braket{\psi_i}{\psi_j}$. 
For a normalized state, the corresponding diagonal entry is $1$, and for subnormalized states it is strictly less than $1$. 
We also note that Gram matrices are always positive semidefinite, and thus $\frac{1}{n} G$ is a trace-bounded positive semidefinite matrix. 
From Lemma~\ref{lem:obvious}, we know thus have $\{ \ket{\psi_1}, \ldots, \ket{\psi_n} \}$ is $k$-learnable if and only if $\frac{1}{n} G \in \Ik$. 

We thus arrive at the following membership problem for learning width.  

\begin{definition}[\lw] 
Given $k$ be a positive integer, an indexed list of (possibly subnormalized) pure quantum states $\{ \ket{\psi_1}, \ldots, \ket{\psi_n} \} \subset \C^n$, and a real parameter $\delta > 0$, solve the $\fw$ problem for its normalized Gram matrix $\frac{1}{n} G$.  
\end{definition} 

Now that we have defined these two decision problems, we can state an immediate corollary. 

\begin{corollary} \label{kinclearn}
$\lw$ and $\fw$ are polynomial-time reducible to each other. 
\end{corollary} 

The proof follows since if given the states, one can compute the Gram matrix to precision $\delta$ and, conversely, if given a trace-bounded positive semidefinite matrix $X$, one can compute vectors $\{ v_1, \ldots, v_n \}$ such that $X_{i,j} = v_i^*v_j$ via the (efficient) Cholesky decomposition algorithm. 
In other words, there is a polynomial-time algorithm to map trace-bounded positive semidefinite operators to an indexed list if subnormalized pure quantum states, and vice versa. 

We now define the optimization versions of the preceding decision problems. 

\begin{definition}[\mfw] 
Given a positive semidefinite matrix $X\in \Pos(\mathbb{C}^n)$ of trace $\leq 1$, find the smallest $k$ such that $(k, X)$ is a YES instance to $\fw$.   
\end{definition} 

\begin{definition}[\mlw]  
Given an indexed list of (possibly subnormalized) pure quantum states $\{ \ket{\psi_1}, \ldots, \ket{\psi_n} \} \subset \C^n$, find the smallest $k$ such that $(k, \{ \ket{\psi_1}, \ldots, \ket{\psi_n} \})$ is a YES instance to $\lw$.   
\end{definition} 

By standard techniques, e.g., binary search, we have that the decision problem
$\fw$ and its optimization variant $\mfw$ are polynomial-time reducible to each
other. Similarly, $\lw$ and $\mlw$ are polynomial-time reducible to each
other~\cite{kleinberg2006algorithm}. Combined with \cref{kinclearn}, we obtain
the following. 

\begin{theorem} \label{thm:reducible} 
The following problems are all polynomial-time reducible to each other: 
\begin{itemize} 
\item $\lw$; 
\item $\mlw$; 
\item $\fw$; 
\item $\mfw$.  
\end{itemize}
\end{theorem} 

To show that $\lw$ is NP-hard, and thus they all are, we examine the (weak) optimization over matrices with factor width at most $k$, next. 


\subsection{Weak optimization problems} 

In order to prove NP-hardness of testing weak membership in $\Ik$, we will use a reduction from \textit{weak optimization} over $\Ik$, which we now define, following~\cite{liu2007complexity}.
 
\begin{definition}[$\WOPT(K, \delta)$]  
The weak optimization problem over a set $K$ is as follows. 
Given a Hermitian matrix $C$ with $\| C \|_{F} = 1$, and real parameters $\gamma$ and $\delta > 0$, decide between the following two cases: 
\begin{itemize}[itemsep=2pt, topsep=8pt]
\item \textup{Yes case:} There exists $X \in S(K,-\delta)$ with $\ip{C}{X} \geq \gamma + \delta$;  
\item \textup{No case:} For all $X \in S(K, \delta)$, we have $\ip{C}{X} \leq \gamma - \delta$. 
\end{itemize} 
\end{definition} 

The reduction from weak optimization to weak membership is given as follows.

\begin{fact}[Theorem 2.3 in \cite{liu2007complexity}]\label{fact}
Given a closed convex set $K \subseteq \Herm(\C^n)$ satisfying 
\begin{equation} 
B(X,r) \subseteq K \subseteq B(0,R), 
\end{equation} 
for some $X \in \Herm(\C^n)$, and real parameters $r, R$ satisfying $R/r \leq {\operatorname{poly}}(n)$. 
For any $\delta$ which is inverse-polynomial in $n$, there exists $\delta'$ which is inverse-polynomial in $n$ for which there is a polynomial-time oracle reduction from $\WOPT(k,\delta)$ to $\WMEM(k,\delta')$.
\end{fact} 

For shorthand, we will use $\operatorname{poly}$ to denote an arbitrary polynomial function, and write this statement simply as $\WOPT(K,1/\operatorname{poly})$ reduces to $\WMEM(K,1/\operatorname{poly})$. We prove that this reduction holds for $\Ik$ in the next section.
  

\section{NP-hardness of {\fw}}  
\label{sec:hardness-proofs}

In this section we prove that deciding weak membership in $\fw$ is NP-hard.

\begin{theorem}\label{thm:nphard}
$\WMEM(\Ik,1/\operatorname{poly})$ is NP-hard. 
\end{theorem} 

Our starting point is a reduction from $\WOPT(\Ik,1/\operatorname{poly})$.

\begin{lemma}\label{lem:reduction}
There is a polynomial-time oracle reduction from the problem $\WOPT(\Ik,1/\operatorname{poly})$ to the problem $\WMEM(\Ik,1/\operatorname{poly})$.
\end{lemma}
This is a straightforward but slightly tedious application of Fact~\ref{fact}.  See~\cref{app:reduction}. Hence, it suffices to prove NP hardness of $\WOPT(\Ik,1/\operatorname{poly})$.

For a Hermitian matrix $C\in\Herm(\mathbb{C}^n)$, consider the linear optimization problem 
\begin{equation} \label{Opt}
\mu(k,C):=\max \{ \ip{C}{X} : X\in \Ik \}. 
\end{equation} 
(Note that the maximum is attained since we are optimizing a continuous function over a compact set.) If a linear optimization problem over a compact set has an optimal solution, then it has an optimal solution at an extreme point.
We now characterize the extreme points of $\Ik$.

\begin{lemma} \label{lem:ex}
The set of extreme points of $\Ik$ is 
\begin{equation} 
P=\{ vv^* : v \text{ has at most } k \text{ non-zero entries and } \norm{v}_2 = 1 \}\cup \{0\}. 
\end{equation} 
\end{lemma} 
This is proven in Appendix~\ref{app:ex}. We can now succinctly write down the optimal value of the optimization problem $\mu(k, C)$. The following corollary is immediate.

\begin{corollary} \label{Batman}
For a Hermitian matrix $C\in \Herm(\mathbb{C}^n)$, we have 
\begin{equation} 
\mu(k,C) = \max_{S \subseteq [n], |S|\leq k} \lambda_{\max}(C_S),
\end{equation} 
where $C_S$ the principal submatrix of $C$ with rows and columns indexed by $S$, and $\lambda_{\max}$ denotes the maximum eigenvalue.
\end{corollary} 

In particular, if $C$ is the adjacency matrix of a simple undirected graph $G$, then Corollary~\ref{Batman} shows that $\mu(k, C)$ is equal to the maximum eigenvalue of the adjacency matrix of an at-most-$k$-node subgraph.
It can be shown that the maximum eigenvalue of a $k \times k$ adjacency matrix is maximized when it is a clique. 
This fact plays a critical component in our  NP-hardness proof of $\WOPT(\Ik,1/\operatorname{poly})$.

\begin{definition}[\problem{Clique}~\cite{karp2009reducibility}]
Given a simple undirected nonempty graph $G$ and a positive integer $k$, decide whether $G$ contains a clique of size $k$.  
\end{definition} 

It is known that \problem{Clique} is NP-hard~\cite{karp2009reducibility}. 
The following proposition shows that \problem{Clique} reduces to $\WOPT(\Ik,1/\operatorname{poly})$ for $\Ik$, completing the proof of Theorem~\ref{thm:nphard}. 

\begin{proposition}\label{prop:graph_clique}
    Let $n \geq k \geq 2$ be integers, let $C$ be the adjacency matrix of a simple nonempty undirected graph $G$ on $n$ vertices with $e$ edges, and let $C'=\frac{1}{\sqrt{2e}}C$ be the normalization in Frobenius norm. Then there exist functions $\gamma=\gamma(e,k,n)\in \mathbb{R}$ and $\delta=1/{\operatorname{poly}}(n)$ for which
    \begin{itemize}
        \item $\mu(k,C')\geq \gamma + \delta$ if $G$ contains a $k$-clique.
        \item $\mu(k,C') \leq \gamma - \delta$ if $G$ does not contain a $k$-clique.
    \end{itemize}
\end{proposition}

See Appendix~\ref{app:hardness} for the proof. The key idea is to show that $\mu(k,C)=k-1$ if $G$ contains a $k$-clique, and $\mu(k,C)\leq k-1-\frac{1}{{\operatorname{poly}}(k)}$ if $G$ does not contain a $k$-clique.

 
\section{Succinct proofs of \fw} 
\label{sect:inNP}

In this section, we show that weak membership is in NP, i.e., $\WMEM(\Ik,1/\operatorname{poly}) \in \operatorname{NP}$, by showing how to find \emph{succinct} proofs. 
For this discussion, it helps to recall the definition of $\Ik$:   

\begin{equation} 
\Ik = \{ X \; : \; 
X = \sum_{i} v_i v_i^*, \text{ where each } 
v_i \text{ has at most $k$ non-zero entries}, \;  
\tr(X) \leq 1 \}. 
\end{equation} 

A first thought to prove that $X \in \Ik$ is to simply provide all of the $v_i$ vectors and check the conditions above. 
However, there are an exponential number of these vectors (at least $\binom{n}{k}$), making this proof too long (it must be polynomial in size). 

The trick to reducing the dimension is notice that $\Herm(\C^n)$ is a vector space of dimension $n^2$, so Carath\'{e}odory's theorem (see, e.g., \cite[Theorem~1.9]{watrous2018theory}) tells us that any member of $\Ik$ can be written as a convex combination of $n^2 + 1$, or fewer, extreme points of that convex set. 
Since we have already characterized the extreme points in Lemma~\ref{lem:ex}, we know that can write any $X \in \Ik$ as 
\begin{equation} \label{check}
X = \sum_{i=1}^{n^2+1} v_i v_i^*
\end{equation}  
where $v_i$ has at most $k$ non-zero entries. 
(We absorbed the convex combination scalars into the definitions of $v_i$.) 
Thus, a short proof of membership in $\Ik$ simply provides the $v_i$ matrices above, checking that they have the proper support, Equation \cref{check}, and its trace.

We remark that even though this proof is succinct, this does not imply it is easy to find. 
Indeed, when optimizing or testing membership of this set, we do not know which polynomially many of the exponentially many $v_i$s are the ones we can use. 


\section{Polynomial-time algorithm for {\fw} for constant k} \label{sect:ellipsoid} 
 
Our approach to efficiently test whether an indexed list of (possibly subnormalized) pure quantum states $\{ \ket{\psi_1}, \ldots, \ket{\psi_n} \}$ is $k$-learnable is to solve the weak membership problem over $\Ik$ of its (scaled) Gram matrix $\frac{1}{n} G$. 

Before continuing, we prove a lemma to bound the Frobenius norm of any matrix we are considering to be in $S(\Ik, \delta)$. 

\begin{lemma} 
Given $X$ and $\delta$, if $\| X \|_{F} > 1 + \delta$, 
then $X \not\in S(\Ik, \delta)$. 
\end{lemma}  

\begin{proof} 
By the reverse triangle inequality, we have 
\begin{equation} 
\| X -  Y \|_F \geq \| X \|_F - \| Y \|_F > \delta  
\end{equation} 
for any $Y \in \Ik$. 
\end{proof} 

Therefore, we have a two-step algorithm. 
First, check if $\| X \|_F > 2 + \delta$.  
If so, return ``No'', $X \not\in S(\Ik, \delta)$. 
Otherwise, solve an SDP which we describe below.  
Given $X$, the matrix we wish to test $\Ik$ membership of, one can solve the following semidefinite program 
\begin{equation} \label{SDP:step1}
    \alpha = \min \{ \| X - Y \|_F^2 : Y \in \Ik \}  
\end{equation} 
to within precision $\delta^2/3$. 
After doing this, it is clear whether there exists $Y \in \Ik$ $\delta$-close to $X$, and thus to respond with ``Yes'' or ``No''. 

Before continuing, we present the well-known semidefinite programming formulation of the Frobenius norm, for completeness. 

Given $X$ (not necessarily square), consider the block matrix 
\begin{align} 
M = \begin{pmatrix}
I & X \\ 
X^* & Z
\end{pmatrix}.   
\end{align} 
Then by looking at Schur complements, we have $M$ positive semidefinite if and only if $Z - X^* X$ is positive semidefinite. 
Thus, minimizing the trace of $Z$ yields $\tr(X^*X) = \| X \|_F^2$. 

Thus, we can write SDP~\cref{SDP:step1} as the SDP below. 

\begin{align} \label{SDP:step2}
\alpha 
& = \min \left\{ \Tr(Z) \; : \: 
\begin{pmatrix}
I & Y-X \\ 
Y-X & Z
\end{pmatrix} \geq 0, \;  
Y \in \Ik  
\right\} \\ 
& = \min \left\{ \Tr(Z) : 
\begin{pmatrix}
I & Y-X \\ 
Y-X & Z
\end{pmatrix} \geq 0,  
\sum_{S : |S| = k} W_S = Y, 
W_S \geq 0, 
\bra{i} W_S \ket{i} = 0, \forall i \notin S, 
\tr(Y) \leq 1 
\right\}. 
\end{align}  

It is helpful to bound the trace of $Z$ without affecting the value of $\alpha$ (so it does not cut off any near optimal solutions). 
We know that in any optimal solution we have 
\begin{equation} 
\tr(Z) = \| X - Y \|_F^2 \leq \| X \|_F^2 \leq (\delta + 2)^2.   
\end{equation} 
where the first inequality holds since $Y = 0$ is a feasible solution. 

Thus, the SDP we actually solve takes the form 
\begin{align} \label{SDP:step3}
\alpha 
& = \min \bigg\{ \Tr(Z) : 
\begin{pmatrix}
I & Y-X \\ 
Y-X & Z
\end{pmatrix} \geq 0,  
\sum_{S : |S| = k} W_S = Y, 
W_S \geq 0, 
\bra{i} W_S \ket{i} = 0, \forall i \notin S, \\ 
&  \quad \quad \quad \quad \quad \quad \quad \quad 
\tr(Y) \leq 1, \tr(Z) \leq (\delta + 2)^2 
\big\}. 
\end{align} 

Assuming we can write the SDP in standard form 
\begin{equation} 
\alpha = \sup \{ \ip{A}{X} \; : \: \Phi(X) = B, X \geq 0 \}, 
\end{equation} 
then one can approximate $\alpha$ to $\epsilon > 0$ precision (additive error) in time $\poly(a, b, \log(R), \log(1/\epsilon), c)$ using the Ellipsoid Algorithm~\cite{grotschel2012geometric}, 
where: 
\begin{itemize} 
\item $X$ is $a \times a$; 
\item $B$ is $b \times b$; 
\item $\| X \|_2 \leq R$ for all feasible $X$ (and there exists feasible $X$); 
\item $c$ is the maximum bit-length of the entries in $A$, $B$, and $J(\Phi)$, where $J$ is the Choi representation of $\Phi$.  
\end{itemize} 

To modify our SDP, we first write the linear operator 
\begin{equation} 
\Psi_1(Y,Z,W) = \begin{pmatrix}
0 & Y \\ 
Y & Z
\end{pmatrix} - W \quad \text{ and } \quad B_1 = 
\begin{pmatrix}
-I & X \\ 
X & 0
\end{pmatrix} 
\end{equation} 
where $Y,Z,W$ are all positive semidefinite. 
Then $\Psi_2(Y,Z,W) = B_1$ models the 
\begin{equation}
\begin{pmatrix}
I & Y-X \\ 
Y-X & Z
\end{pmatrix} \geq 0 
\end{equation} 
constraint. 
If we define 
\begin{equation} 
\Psi_2(Y, W_1, \ldots, W_{\binom{n}{k}})  = \sum_{i=1}^{\binom{n}{k}} W_i - Y \quad \text{ and } \quad B_2 = 0  
\end{equation} 
this models the 
$Y = \sum_{S:|S|=k} W_S$ constraint. 
If we define 
\begin{equation} 
\Phi_S(W_S)  = \sum_{i \not\in S} \bra{i} W_S \ket{i} 
\quad \text{ and } \quad
\tilde{B}_S = 0  
\end{equation} 
this models the 
$\bra{i} W_S \ket{i} = 0$ constraints. 
If we define 
\begin{equation} 
\Psi_3(Y, t)  = \ip{I}{Y} + t
\quad \text{ and } \quad 
B_3 = 1
\end{equation} 
where $t \geq 0$  
this models the 
$\tr(Y) \leq 1$ constraint. 
If we define 
\begin{equation} 
\Psi_4(Z, s) = \ip{I}{Z} + s \quad \text{ and } \quad B_4 = (\delta+2)^2 
\end{equation} 
for $s \geq 0$  
this models the $\tr(Z) \leq (\delta + 2)^2$ constraint. 

Therefore, we have the variables 
$(Y, Z, W, W_1, \ldots, W_{\binom{n}{k}}, t, s)$, all positive semidefinite. 
Also, we have the constants $B_1, B_2, B_3, B_4, \tilde{B}_1, \ldots, \tilde{B}_{\binom{n}{k}}$. 
By defining a single SDP variable as 
\begin{equation} \label{big_primal_var}
Y \oplus Z \oplus W \oplus W_1 \oplus t \oplus s \oplus W_1 \oplus \cdots \oplus W_{\binom{n}{k}} 
\end{equation} 
and a single constant matrix as 
\begin{equation} 
B_1 \oplus B_2 \oplus B_3 \oplus B_4 \oplus \tilde{B}_1 \oplus \cdots \oplus \tilde{B}_{\binom{n}{k}} 
\end{equation} 
and a single linear map $\Psi$ as 
\begin{equation} 
\Psi(Y, Z, W, W_1, \ldots, W_{\binom{n}{k}}, t, s) 
= \Psi_1 \oplus \Psi_2 \oplus \Psi_3 \oplus \Psi_4 \oplus \Phi_1 \oplus \cdots \oplus \Phi_{\binom{n}{k}} 
\end{equation} 
we can write the affine SDP constraints as $\Psi(X) = B$. 
Lastly, define $A$ as 
\begin{equation} 
0 \oplus I \oplus 0 \oplus 0 \oplus 0 \oplus \cdots \oplus 0 \oplus 0  
\end{equation} 
(where the dimensions of the $0$ matrices align with Equation~\cref{big_primal_var}) finishes putting the SDP in standard form. 

First notice that the bit-length of the data of the standard form SDP is the same as $X$ and $\delta$. 

It is now straightforward to check that we can take 
\begin{itemize} 
\item  $R = \poly (n, \delta) = \poly(n)$  when $\delta = 1/\poly(n)$.   
\item $a = \poly(n)$ when $k$ is constant. 
\item $b = \poly(n)$. 
\item $\epsilon = \delta^2/3$ (this suffices to tell if there is a $Y \in \Ik$ such that $\| X - Y \|_F \leq \delta$ or not).  
\end{itemize} 
Thus, the SDP runs in time $\poly(n)$ and the bit-length of $X$ and $\delta$.   


\section{Polynomial-time algorithm for {\fw} when the rank is constant} 
\label{sec:low_rank} 

In this section, we give a $\operatorname{poly}(n)$-time algorithm for computing the factor width of a matrix $M \in \Pos(\mathbb{C}^n)$ when the rank of $M$ is fixed (and $k$ is arbitrary). Note that this fact is trivial in the rank~$1$ case (and is exploited considerably in the rank~$1$ case in~\cite{harrow2025randomized}).

\begin{theorem}\label{thm:low_rank_easy}
    Let $M \in \Pos(\C^n)$ has rank $r$, and $k$ is a positive integer. Then there exists a positive integer $m \in O(n^r)$ and subspaces $T_1, \ldots, T_m \subseteq \C^n$ (which can be found in $\operatorname{poly}(n)$ time) with the property that $M$ is $k$-incoherent if and only if there exist matrices $M_1, \ldots, M_m \in \Pos(\C^n)$ for which
    \begin{align}\label{eq:low_rank_sdp}
        M = \sum_{j=1}^m M_j \quad \text{and} \quad M_j = \Pi_j M_j \Pi_j
    \end{align}
    where, for all $j \in \{1,2,\ldots,m\}$, $\Pi_j$ is the orthogonal projection onto $T_j$.
\end{theorem}

We note that determining the existence of $M_1,...,M_m$ can be solved by semidefinite programming, and approximated in a similar way as was done in Section~\ref{sect:ellipsoid}. 

Before proving this theorem, we note that it immediately implies a $\operatorname{poly}(n)$-time algorithm for determining $k$-learnability of $n$ pure states in any constant dimension $d$. This follows from the fact that the rank of a Gram matrix is exactly the minimum dimension of states $\ket{\psi_1}$, $\ldots$, $\ket{\psi_n}$ with that Gram matrix.

\begin{corollary}\label{cor:learning_easy_small_dim}
    Suppose $S = \{ \ket{\psi_1}, \ldots, \ket{\psi_n} \} \subset \C^d$ and $d$ is constant. Then the learning width of $S$ can be determined in $\operatorname{poly}(n)$ time by solving a polynomial number of polynomial-sized linear systems, and a polynomial-sized semidefinite program.
\end{corollary}

Now we prove Theorem~\ref{thm:low_rank_easy}.

\begin{proof}[Proof of Theorem~\ref{thm:low_rank_easy}]
    Let $S := \range(M) \subseteq \C^n$, so $\dim(S) = r$. For each $j \in \{1, 2, \ldots, n\}$, consider the subspace
    \begin{equation}
        S_j := S \cap \{ {v} \in \C^n : v(j) = 0 \}.
    \end{equation}
    Similarly, for an index set $J \subseteq \{1, 2, \ldots, n\}$, define
    \begin{equation}
        S_J := \bigcap_{j \in J} S_j = S \cap \{ {v} \in \C^n : v(j) = 0 \ \ \forall \ j \in J\}.
    \end{equation}
    Let $\{ T_1, T_2, \ldots, T_m\}$  be the set of such subspaces corresponding to all $J \subseteq \{1, 2, \ldots, n\}$ with $|J| = n - k$. Our first goal is to show that $m \leq \sum_{j=0}^r\binom{n}{j} \in O(n^r)$ and that these subspaces can be found in polynomial time.

    To this end, we proceed recursively. Start with the set $R_0 := \{S\}$. Then, for each integer $j \in \{1, 2, \ldots, n-k\}$, we define
    \begin{equation}
        R_j := \big\{ R \cap \{ {v} \in \C^n : v(i) = 0 \} : R \in R_{j-1}, \ i \in \{1, 2, \ldots, n\}\big\}.
    \end{equation}
    For each $R \in R_{j-1}$ and $i \in \{1, 2, \ldots, n\}$, we have either $R \cap \{ {v} \in \C^n : v(i) = 0 \} = R$ or $\dim(R \cap \{ {v} \in \C^n : v(i) = 0 \}) = \dim(R) - 1$. Since $\dim(S) = r$, this recursive construction cannot decrease the dimension of $S$ more than $r$ times, so the total number of distinct subspaces in $R_j$ can never exceed $\sum_{j=0}^r\binom{n}{j}$. Since $R_{n-k} = \{ T_1, T_2, \ldots, T_m\}$, the fact that $m \leq \sum_{j=0}^r\binom{n}{j}$ follows.

    To see that (orthonormal bases of) $T_1, T_2, \ldots, T_m$ can be computed in $\operatorname{poly}(n)$ time, simply notice that (as noted earlier) each $R_{j-1}$ contains at most $\sum_{j=0}^r\binom{n}{j}$ subspaces, so each $R_j$ can be constructed by computing orthonormal bases for at most $n\sum_{j=0}^r\binom{n}{j}$ subspaces. Since the recursion runs for $n-k$ steps, this amounts to computing at most $(n-k)n\sum_{j=0}^r\binom{n}{j} \leq n^2\sum_{j=0}^r\binom{n}{j}$ subspaces, which is polynomial in $n$.

    Our next goal is to show that $M$ is $k$-incoherent if and only if there exist matrices $M_1, \ldots, M_m \in \Pos(\C^n)$ for which~\eqref{eq:low_rank_sdp} holds. To this end, note that if $M$ can be written in this form then it has factor width at most $k$, since every vector in $T_j$ has at most $k$ non-zero entries, so each $M_j = \Pi_j M_j \Pi_j$ equals $0$ outside of a single $k \times k$ principal submatrix. Conversely, if $M$ cannot be written in this form then, since $\bigcup_{j=1}^m T_j$ contains every vector ${v} \in S$ with at most $k$ non-zero entries, $M$ must have factor width strictly larger than $k$.
\end{proof}

\begin{example}\label{exam:rank_2_fw}    
    To illustrate how Theorem~\ref{thm:low_rank_easy} works, suppose we want to determine the factor width of the rank-$3$ matrix
    \begin{equation}
        M = \begin{bmatrix}
         2 & 1 & 1 & -1 \\
         1 & 2 & 0 & 1 \\
         1 & 0 & 2 & -1 \\
         -1 & 1 & -1 & 2
        \end{bmatrix}.
    \end{equation}
    We start by finding a basis for $S := \range(M)$, which can be done by picking a linearly independent set of $r = 3$ columns of $M$: $S = \operatorname{span}\{(2,1,1,-1), (1,2,0,1), (1,0,2,-1)\}$. Then $R_0 = \{S\}$ and we proceed recursively as in the proof of Theorem~\ref{thm:low_rank_easy}:
    \begin{align*}
        R_1 = \{S_1, S_2, S_3, S_3\}, \quad \text{where} \quad S_1 & = \operatorname{span}\{(0,1,-1,1), (0,1,-3,1)\}, \\
        S_2 & = \operatorname{span}\{(1,0,2,-1), (3,0,2,-3)\}, \\
        S_3 & = \operatorname{span}\{(1,2,0,1), (3,2,0,-1)\}, \ \ \text{and} \\
        S_4 & = \operatorname{span}\{(1,1,1,0), (3,3,1,0)\}.
    \end{align*}
    To determine whether or not $M$ is $3$-incoherent, we let $\Pi_1$, $\Pi_2$, $\Pi_3$, and $\Pi_4$ be the orthogonal projections onto $S_1$, $S_2$, $S_3$, and $S_4$, respectively. We then use semidefinite programming to determine whether or not there exist matrices $M_1, M_2, M_3, M_4 \in \Pos(\C^4)$ for which
    \begin{equation}
        M = M_1 + M_2 + M_3 + M_4, \quad \text{and} \quad M_j = \Pi_j M_j \Pi_j \quad \text{for all} \quad j \in \{1,2,3,4\}.
    \end{equation}
    Indeed, such matrices do exist:
    \begin{equation}
        M_1 = \begin{bmatrix}
            0 & 0 & 0 & 0 \\
            0 & 1 & -1 & 1 \\
            0 & -1 & 1 & -1 \\
            0 & 1 & -1 & 1
        \end{bmatrix}, \ M_2 = \begin{bmatrix}
            1 & 0 & 0 & -1 \\
            0 & 0 & 0 & 0 \\
            0 & 0 & 0 & 0 \\
            -1 & 0 & 0 & 1
        \end{bmatrix}, \ M_3 = \begin{bmatrix}
            0 & 0 & 0 & 0 \\
            0 & 0 & 0 & 0 \\
            0 & 0 & 0 & 0 \\
            0 & 0 & 0 & 0
        \end{bmatrix}, \ M_4 = \begin{bmatrix}
            1 & 1 & 1 & 0 \\
            1 & 1 & 1 & 0 \\
            1 & 1 & 1 & 0 \\
            0 & 0 & 0 & 0
        \end{bmatrix},
    \end{equation}
    so $M$ is $3$-incoherent.
    
    To similarly determine whether or not $M$ is $2$-incoherent, we proceed further with the recursive construction by computing
    \begin{align*}
        R_2 = \{S_{\{1,2\}}, S_{\{1,3\}}, S_{\{2,3\}}, S_{\{3,4\}}\}, \quad \text{where} \quad S_{\{1,2\}} = S_{\{1,4\}} = S_{\{2,4\}} & = \operatorname{span}\{(0,0,1,0)\}, \\
        S_{\{1,3\}} & = \operatorname{span}\{(0,1,0,1)\}, \\
        S_{\{2,3\}} & = \operatorname{span}\{(1,0,0,-1)\}, \ \ \text{and} \\
        S_{\{3,4\}} & = \operatorname{span}\{(1,1,0,0)\}.
    \end{align*}
    It follows that the only vectors in $\range(M)$ with $k = 2$ or fewer non-zero entries are the scalar multiples of ${v_1} := (0,0,1,0)$, ${v_2} := (0,1,0,1)$, ${v_3} := (1,0,0,-1)$, and ${v_4} := (1,1,0,0)$, so $M$ is $2$-incoherent if and only if there exist non-negative real scalars $c_1$, $c_2$, $c_3$, and $c_4$ for which
    \begin{equation}
        M = c_1{v_1}{v_1}^* + c_2{v_2}{v_2}^* + c_3{v_3}{v_3}^* + c_4{v_4}{v_4}^*.
    \end{equation}
    It is straightforward to use semidefinite programming (or even just solve by hand in this small example) to see that no such scalars exist, so $X$ is not $2$-incoherent. It follows that $X$ has factor width $3$.
\end{example}


\subsection*{Software} 
Companion software implements numerical routines for determining $k$-incoherence of matrices, the factor width, and $k$-learnability of quantum state collections. These are provided within the \texttt{toqito} quantum information package~\cite{russo2021toqito}, leveraging \texttt{PICOS}~\cite{sagnol2022picos} and the \texttt{CVXOPT} solver~\cite{andersen2020cvxopt} for semidefinite program optimization.


\subsection*{Acknowledgements} 
The authors thank Angus Lowe and Richard Cleve for helpful conversations. 
N.J.\ was supported by NSERC Discovery Grant number RGPIN-2022-04098. 
J.S.\ is funded in part by the Commonwealth of Virginia's Commonwealth Cyber Initiative (CCI) under grant number 469351. 

\bibliographystyle{unsrt}
\bibliography{bib}


\appendix 

\crefalias{section}{appendix}
\crefalias{subsection}{appendix}
\crefalias{subsubsection}{appendix}

 
\section{Proof of Lemma~\ref{lem:SDP_same}} \label{sdp_proof}  

We reproduce SDP~\cref{SDP_OG}, below  
\begin{equation} \label{SDP_OG_rep}
\alpha = \min\left\{ \frac{1}{n} \sum_{i=1}^n \bra{\psi_i} \left(\sum_{S \, : \, i \notin S} M_S\right) \ket{\psi_i} \right\} 
\end{equation} 
where the minimization is over POVMS. 

Define $r_i = \| \ket{\psi_i} \|_2$ and 
\begin{equation} 
\ket{\phi_i} = \left\{ 
\begin{array}{rcl}
(1/r_i) \ket{\psi_i} & \text{ when } r_i \neq 0 \\ 
\textup{To be determined} & \text{ when } r_i = 0 
\end{array} 
\right. 
\end{equation} 
to make $\ket{\psi_i} = r_i \ket{\phi_i}$. 
Notice that we have the freedom to choose states above, and we do this such that each one is orthogonal to everything else. 
We can do this by increasing the dimension if need be and this particular choice will ultimately make the corresponding Gram matrix nice. 
For an index $i$ and subset $S$, define 
\begin{equation} 
t_{i,S} = \left\{ 
\begin{array}{rcl}
0 & \text{ when } i \in S \\ 
1 & \text{ when } i \not\in S.  
\end{array} 
\right. 
\end{equation} 
Lastly, define $q_i = 1/n$ and $R_{i,S} = r_i^2 t_{i,S}$, for all $i$ and $S$. 

With these things defined, we can write the SDP in the following way   
\begin{equation} \label{SDP_OG_rep2}
\alpha = \min\left\{ \sum_{i=1}^n \sum_{S : |S| = k} q_i R_{i,S} \ip{M_S}{ \kb{\psi_i} } \right\} 
\end{equation} 
where the minimization is over POVMS (the explicit constraints are omitted for clarity).  
When written in this way, the objective function can be interpreted as the average \emph{reward} when given $\ket{\psi_i}$ with probability $q_i$, and measured with the POVM $(M_S : |S| = k)$ to receive reward $R_{i,j}$. 
This framework has been studied in~\cite{mohan2023generalized} and it can be shown that the optimal value of the preceding SDP equals that of the SDP  
\begin{equation} 
\alpha = \min\left\{ \sum_{i=1}^n \sum_{S : |S| = k} q_i R_{i,S} \ip{W_S}{ \kb{i} } : \sum_{S : |S| = k} W_S = \tilde{G}, W_S \geq 0 \right\} 
\end{equation} 
where we used $\tilde{G}$ as the Gram matrix for the (normalized) states $\{ \ket{\phi_1}, \ldots, \ket{\phi_n} \}$. 
 
By substituting most of the parameters back in, we get 
\begin{align} 
\alpha  
& = \min\left\{ \frac{1}{n} \sum_{i=1}^n r_i^2 \bra{i} \left( \sum_{i \not\in S} W_S \right) \ket{i} : \sum_{S : |S| = k} W_S = \tilde{G}, W_S \geq 0 \right\} \\ 
& = \min\left\{ \sum_{i=1}^n r_i^2 \bra{i} \left( \sum_{i \not\in S} \tilde{W}_S \right) \ket{i} : \sum_{S : |S| = k} \tilde{W}_S = \frac{1}{n} \tilde{G}, \tilde{W}_S \geq 0 \right\} 
\end{align} 

The rest of the analysis is easier if we define a diagonal matrix and its (generalized) inverse. 
If we define the diagonal matrix $D_{i,i} = {r_i}$, then we can rewrite ${r_i} \ket{i} = D \ket{i}$ and also the Gram matrix $G$ of the original states as 
\begin{equation} \label{eq:johnwick} 
G_{i,j} = \braket{\psi_i}{\psi_j} = r_i r_j \braket{\phi_i}{\phi_j} = (D \tilde{G} D^*)_{i,j}.  
\end{equation}  
Define the Moore-Penrose generalized inverse of $D$ as $D^+$ where $(D^+)_{i,j} = 1/r_i$ when $r_i \neq 0$, and $0$ when $r_i = 0$. 
Then we have, for $i,j $ such that $r_i, r_j \neq 0$: 
\begin{equation} \label{eq:johnwick2} 
\tilde{G}_{i,j} = \braket{\phi_i}{\phi_j} = \frac{1}{r_i r_j} \braket{\psi_i}{\psi_j} = (D^{+} G D^{+})_{i,j}.  
\end{equation} 
From this, it is easy to see by the choice of $\ket{\phi_i}$ when $r_i = 0$ that 
\begin{equation} \label{eq:nice}
\tilde{G} = D^{+} G D^{+} + I_T, 
\end{equation} 
where $T$ only has support on the entries where $r_i = 0$. 
 
With this, the SDP reduces to 
\begin{align} \label{SDP:a}
\alpha 
& = \min\left\{ \sum_{i=1}^n \bra{i} \left( \sum_{i \not\in S} D \tilde{W}_S D \right) \ket{i} : \sum_{S : |S| = k} \tilde{W}_S = \frac{1}{n} \tilde{G}, \tilde{W}_S \geq 0 \right\}. 
\end{align} 

To finish the proof, we now argue that $\alpha$ is also equal to the optimal objective function value of SDP~\cref{SDP_red} below, which we have defined as $\beta$. 
\begin{equation} \label{SDP:b} 
\beta = \min \left\{ \sum_{i=1}^n \bra{i} \sum_{S: i \notin S} W_S \ket{i} : \sum_{S : |S| = k} W_S = \frac{1}{n} G, \; W_S \geq 0 \right\} 
\end{equation} 

\paragraph{Proof that $\beta \leq \alpha$.} 
Suppose $(\tilde{W}_S : |S| = k)$ is feasible in the $\alpha$ SDP~\cref{SDP:a} with objective function value $t$. 
Define the following $W_S := D \tilde{W} D^* \geq 0$. 
Since $(\tilde{W}_S : |S| = k)$ is feasible, we know that $\sum_{S : |S| = k} \tilde{W}_S  = (1/n) \tilde{G}$, and thus  
\begin{equation} 
\sum_{S : |S| = k} W_S 
= \sum_{S : |S| = k} D \tilde{W}_S D^* 
= (1/n) D \tilde{G} D^* 
= (1/n) G. 
\end{equation} 
Thus $(W_S : |S| = k)$ is feasible in the $\beta$ SDP~\cref{SDP:b} with objective function value 
\begin{equation} 
\sum_{i=1}^n \bra{i} \sum_{S: i \notin S} W_S \ket{i} 
= 
\sum_{i=1}^n \bra{i} \sum_{S: i \notin S} D \tilde{W}_S D \ket{i} = t. 
\end{equation} 
Thus, $\beta \leq \alpha$. 

\paragraph{Proof that $\alpha \leq \beta$.}  

Let $(W_S : |S| = k)$ be feasible in the $\beta$ SDP~\cref{SDP:b} with objective function value $s$.  
Define $\tilde{W}_S := D^{+} W_S D^{+} \geq 0$. 
Since $(W_S : |S| = k)$ is feasible, we have that 
$\sum_{S : |S| = k} W_S = (1/n)G$. 
Thus we have 
\begin{equation} 
\sum_{S : |S| = k} \tilde{W}_S 
= \sum_{S : |S| = k} D^{+} W_S D^{+} 
= \sum_{S : |S| = k} (1/n) D^{+} G D^{+} 
\end{equation} 
which is \emph{almost feasible}. 
If we add $(1/n) I_T$ to any of the $W_S$, it will still be positive semidefinite and now $(\tilde{W}_S : |S| = k)$ is feasible in the $\alpha$ SDP~\cref{SDP:a}. 
Note that this addition of $I_T$ does not contribute to the objective function value since $D I_T D = 0$, so we are not concerned to which variable it is added. 

All that remains is to compute the objective function value of this feasible solution. 
Notice that since $G_{i,i} = 0$ for any $i$ such that $r_i = 0$, it is also the case that $(W_S)_{i,i} = 0$ for all $S$. 
We can now compute the objective function value as 
\begin{equation} 
\sum_{i=1}^n \bra{i} \sum_{S: i \notin S} D \tilde{W}_S D \ket{i} 
= 
\sum_{i=1}^n \bra{i} \sum_{S: i \notin S} D D^{+} W_S D^{+} D \ket{i} 
= 
\sum_{i:r_i \neq 0} \bra{i} \sum_{S: i \notin S} W_S \ket{i} 
= 
\sum_{i} \bra{i} \sum_{S: i \notin S} W_S \ket{i} = s. 
\end{equation} 
Thus, we have $\alpha \leq \beta$ as required. 


\section{Proof of Lemma~\ref{lem:obvious}}\label{app:equiv}

Here we prove \cref{lem:obvious} which states that an indexed list of (possibly sub-normalized) pure quantum states $\{ \ket{\psi_1}, \ldots, \ket{\psi_n} \} \subset \C^d$ is $k$-\emph{learnable} if and only if its Gram matrix is $k$-incoherent.  

\paragraph{``If'' direction.} 
Let $G$ be the Gram matrix and suppose it is $k$-incoherent. 
Then SDP~\cref{SDP_red} has objective value $0$, and by \cref{lem:SDP_same}, so does SDP~\cref{SDP_OG}. 
Since that SDP attains an optimal solution as the feasible region is the set of valid POVMS with $\binom{n}{k}$ outcomes, and thus compact, we know the existence of a POVM which satisfies the conditions for $k$-learnability. 

\paragraph{``Only if'' direction.} 
Suppose the set is $k$-learnable. 
Then SDP~\cref{SDP_OG} has objective function $0$, and by \cref{lem:SDP_same}, so does SDP~\cref{SDP_red}. 
If SDP~\cref{SDP_red} attains an optimal solution, then said solution asserts that $G$ is $k$-incoherent. 
Thus, showing attainment of an optimal solution suffices to finish the proof. 

We proceed with a similar argument as above, that the feasible region is bounded, and thus compact. 
Since each $W_S$ is positive semidefinite, we have that $\| W_S \|_{\tr} = \tr(W_S) \leq \frac{1}{n} \tr(G) \leq 1$. 
Thus, the feasible region is bounded in the trace-norm (and hence any norm since we are working in finite dimensions).  


\section{Proof of Lemma~\ref{lem:reduction}}\label{app:reduction}
Here we prove \cref{lem:reduction}, which states that the problem $\WOPT(\Ik,1/\operatorname{poly})$ reduces to the problem $\WMEM(\Ik,1/\operatorname{poly})$. 
We require a lemma and a corollary to prove this result. 

\begin{lemma}\label{lem:abs_k_incoh_ball_here}
    Let $M\in \Herm(\mathbb{C}^n)$. If $\| M - \frac{I}{n} \|_{\textup{F}} \leq 1/d$ then $M$ is a $2$-incoherent positive operator (and thus k-incoherent for all $k \in \{2, \ldots, n\}$). 
\end{lemma}

In particular, this lemma demonstrates that $\Ik$ has non-empty interior.

\begin{proof}
    Multiplying by $n$, it suffices to prove that if $\norm{M-I}_F \leq 1$, then $M$ is 2-incoherent. It is easy to see that $M$ must be positive semidefinite under this condition. Observe that
    \begin{align}\label{eq:frob_norm_from_eig}
        \|M - I\|_{\textup{F}}^2 = \sum_{j=1}^n (\lambda_j - 1)^2 = n - 2\tr(M) + \sum_{j=1}^n \lambda_j^2.
    \end{align}
    
    We know from \cite[Theorem~7]{johnston2022absolutely} that if $M \in \Pos(\C^n)$ has $\tr(M) = 1$ and eigenvalues $\lambda_1$, $\lambda_2$, $\ldots$, $\lambda_n$ then $\sum_{j=1}^n \lambda_j^2 \leq 1/(n-1)$ implies that $M$ is $2$-incoherent. By rescaling $M$, we see that this is equivalent (for all $M \in \Pos(\C^n)$, not just those with $\tr(M) = 1$) to $\sum_{j=1}^n \lambda_j^2 \leq (\tr(M))^2/(n-1)$ implying $2$-incoherence of $M$. When combined with Equation~\eqref{eq:frob_norm_from_eig}, this tells us that if
    \begin{align}\label{eq:frob_norm_from_eig_b}
        \|M - I\|_{\textup{F}}^2 \leq n - 2\tr(M) + \frac{(\tr(M))^2}{d-1}
    \end{align}
    then $M$ is $2$-incoherent (and thus $k$-incoherent for all $k \in \{2, \ldots, n\}$).

    If we treat the right-hand-side of Inequality~\eqref{eq:frob_norm_from_eig_b} as a function of $\tr(M)$ then standard calculus techniques show that it is minimized when $\tr(M) = n-1$, and it minimum value is
    \begin{equation} 
        n - 2(n-1) + \frac{(n-1)^2}{n-1} = 1.
    \end{equation} 
    It follows that if $\|M - I\|_{\textup{F}}^2 \leq 1$ then Inequality~\eqref{eq:frob_norm_from_eig_b} must hold, so $M$ must be $2$-incoherent. Taking the square root of both sides of this inequality gives the statement of the theorem.
\end{proof}
  
\begin{corollary}
For any unit vector $u\in \mathbb{C}^n$ supported on at most $k$ coordinates, and any $x \geq 0$, it holds that
\begin{align}\label{eq:dx}
D_x:=(1-\frac{\delta}{n})(x \frac{I}{n} + (1-x) uu^*) \in S(\Ik,-\delta)
\end{align}
for any $\delta\leq \frac{x}{2n}$.
\end{corollary}
\begin{proof}
Let $D_x+\delta E \in S(D, \delta)$, where $E$ is an arbitrary Hermitian matrix with $\norm{E}_F=1$. Note that $\tr(E) \leq n$, so the trace of $D_x+\delta E$ is at most 1. We need to prove that $D_x+\delta E \in \Ik$. Note that $D_x+\delta E = x ((1-\frac{\delta}{d}) \frac{I}{n} + \frac{\delta}{x} E) + (1-\frac{\delta}{n})(1-x) uu^*$. If $\delta \leq \frac{1}{2} n^{-1} x$, then $\delta (x^{-1}+1) \leq n^{-1}$, so by Lemma~\ref{lem:abs_k_incoh_ball_here} the matrix $ ((1-\frac{\delta}{n}) \frac{I}{n} + \frac{\delta}{x} E)  $ is positive semidefinite of factor width $\leq k$. Since $uu^*$ also has factor width $\leq k$, and $\Ik$ is convex, the statement follows.
\end{proof}

\begin{proof}[Proof of Lemma~\ref{lem:reduction}]
The above corollary demonstrates taht $B(I/n, n^{-1}) \subseteq \Ik \subseteq B(0, 1).$ Since the ratio of the radii of these two balls is polynomial, by Fact~\ref{fact},
$\WOPT(\Ik,1/\operatorname{poly})$ reduces to $\WMEM(\Ik,1/\operatorname{poly})$. This completes the proof.
\end{proof}


\section{Proof of Lemma~\ref{lem:ex}}\label{app:ex}

Lemma~\ref{lem:ex} states that the set of extreme points of $\Ik$ is 
\begin{equation} 
P=\{ vv^* : v \text{ has at most } k \text{ non-zero entries and } \norm{v}_2 = 1 \}\cup \{0\}. 
\end{equation} 

For the inclusion $\supseteq$, it is clear that $0$ is an extreme point. Now, let $vv^*$ be another point in $P$, and suppose that $vv^*=p \rho + (1-p) \sigma$ for some $\rho, \sigma \in \Ik$. Then $\rho= \alpha vv^*$ and $\sigma= \beta vv^*$ for some $\alpha, \beta \in [0,1]$ (otherwise $vv^*$ would have rank $>1$, a contradiction). Taking the trace, we have $p \alpha + (1-p) \beta =1$, so $\alpha=\beta=1$.

For the inclusion $\subseteq$, suppose that $\rho$ is an extreme point of $\Ik$. Let
\begin{align}
\rho=\sum_{i=1}^m v_i v_i^*= \sum_{i=1}^m \frac{\tr(\rho)}{\norm{v_i}^2} (v_i v_i^*)
\end{align}
be a factor width $k$ decomposition of $\rho$, as in~\eqref{eq:k_incoh}. Since this is a convex combination of elements of $\Ik$, it must hold that $m=1$, i.e. $\rho=v_1 v_1^*$. It remains to prove that $\tr(\rho)\in \{0,1\}$. Suppose toward contradiction that $\tr(\rho)=\alpha\in (0,1)$. Then $\rho=\frac{1}{2} (\frac{1}{\alpha} \rho) + \frac{1}{2} ((2\alpha-1) \rho)$ is an expression of $\rho$ as a non-trivial convex combination of elements of $\Ik$, hence $\rho$ is not extremal. This completes the proof. 


\section{Proof of Proposition~\ref{prop:graph_clique}}\label{app:hardness}

Here we prove~\cref{prop:graph_clique}. For convenience, we recall the statement of the proposition.

\begin{proposition}[\cref{prop:graph_clique}]
    Let $n \geq k \geq 2$ be integers, let $C$ be the adjacency matrix of a simple nonempty undirected graph $G$ on $n$ vertices with $e$ edges, and let $C'=\frac{1}{\sqrt{2e}}C$ be the normalization in Frobenius norm. Then there exist functions $\gamma=\gamma(e,k,n)\in \mathbb{R}$ and $\delta=1/{\operatorname{poly}}(n)$ for which
    \begin{itemize}
        \item $\mu(k,C')\geq \gamma + \delta$ if $G$ contains a $k$-clique.
        \item $\mu(k,C') \leq \gamma - \delta$ if $G$ does not contain a $k$-clique.
    \end{itemize}
\end{proposition}

\begin{proof}
By Corollary~\ref{Batman}, $\mu(k, C)$ is equal to the maximum eigenvalue of the adjacency matrix of an at-most-$k$-vertex subgraph of $G$. It is a standard fact that the maximum eigenvalue of the adjacency matrix of a $k$-vertex graph $H$ with $f$ edges is upper bounded by $\sqrt{\frac{2f(k-1)}{k}}$, with equality if and only if $H$ is the complete graph (see e.g.~\cite{4138764}). Note that this quantity is equal to $k-1$ if $H$ is the complete graph, and is upper-bounded by
\begin{equation} 
\sqrt{(k-1)^2-2+\frac{2}{k}}\leq k-1-\frac{1}{\epsilon(k)}
\end{equation}
otherwise, for some $\epsilon(k)=1/{\operatorname{poly}}(k)$.

Suppose that $G$ contains a $k$-clique, and let $u\in \mathbb{C}^n$ be a unit vector supported on at most $k$ coordinates for which it holds that $\tr(C' uu^*)=\frac{k-1}{\sqrt{2e}}$ (such a vector exists by extremality, see Lemma~\ref{lem:ex}). Keeping $\delta$ indeterminate for now, let $x=2n \delta$ and let $D_x \in S(\Ik, -\delta)$ as in~\eqref{eq:dx}. Note that
\begin{equation} 
\tr(C'D_x)= (1-\frac{\delta}{n})\frac{(1-x)(k-1)}{e^2} \geq \frac{(1-(2n+1) \delta)(k-1)}{\sqrt{2e}}.
\end{equation} 

Now suppose that $G$ does not contain a $k$-clique. Let $X \in S(\Ik,\delta)$, so $X=F+\delta E$ for some $F \in \Ik$ and Hermitian matrix $E$ of Frobenius norm 1. Then
\begin{equation} 
\tr(C' X) = \tr(C' F) + \delta \tr(C' E) \leq \frac{k-1}{\sqrt{2e}}-\frac{1}{\epsilon \sqrt{2e}}+\delta.
\end{equation}
Let $\gamma=\frac{k-1}{\sqrt{2e}}-\frac{1}{\epsilon \sqrt{2e}}+2 \delta$. Then $\Tr(C' X) \leq \gamma-\delta$. Now, we just need to show that we can choose $\delta = 1/ {\operatorname{poly}}(n)$ small enough so that $\tr(C' D_x) \geq \gamma + \delta$. One can easily verify that the choice
\begin{equation} 
\delta = \frac{1}{\epsilon \sqrt{2e} (3+ \frac{(2n+1)(k-1)}{\sqrt{2e}})}\geq \frac{1}{{\operatorname{poly}}(n)}
\end{equation} 
suffices. 
This completes the proof.
\end{proof} 

\end{document}